\documentclass[11pt]{article}

 \usepackage{tikz}
\usepackage{amsfonts,amssymb,amsmath,amsthm,sectsty,url}
\usepackage[letterpaper,hmargin=.90in,vmargin=.90in]{geometry}
\usepackage{graphicx}
\usepackage{ifpdf}
\usepackage{multirow}
\usepackage{multicol,color,xcolor}
\usepackage{amsbsy}
\usepackage{ragged2e}

   \usepackage[letterpaper=true,pdftex,
    bookmarks,
    plainpages=false, 
    pdfpagelabels=true, 
    colorlinks=true,breaklinks=true,linkcolor=red,citecolor=green,
    filecolor=black,menucolor=black,pagecolor=black,urlcolor=blue
    pdftitle=UW\ LaTeX\ Thesis\ Template, 
    pdfauthor=Jalaj Upadhyay	
    ]{hyperref}{}

\theoremstyle{newstyle}

\newtheorem{theorem}{Theorem}[section]

\newtheorem{fact}[theorem]{Fact}
\newtheorem{corollary}[theorem]{Corollary}

\newtheorem{definition}{Definition}

\newtheorem{construction}[theorem]{Construction}

\newcommand{\thmref}[1]{Theorem~\ref{thm:#1}}

\newcommand{\corref}[1]{Corollary~\ref{cor:#1}}

\newcommand{\secref}[1]{Section~\ref{sec:#1}}

\newcommand{\figref}[1]{Figure~\ref{fig:#1}}

\newcommand{\eqnref}[1]{equation~(\ref{eq:#1})}

\newcommand{\F}{\mathbb{F}}

\newcommand{\p}{\mathsf{Pr}}

\newcommand{\by}{\textbf{y}}
\newcommand{\bx}{\textbf{x}}
\newcommand{\bz}{\textbf{z}}

\newcommand{\cA}{\mathcal{A}}

\newcommand{\cM}{\mathcal{M}}

\newcommand{\cP}{\mathcal{P}}

\newcommand{\cR}{\mathcal{R}}

\newcommand{\prover}{\mathsf{Prover}}

\newcommand{\dist}{\mathsf{dist}}

\newcommand{\brak}[1]{{\langle {#1} \rangle}}
\newcommand{\set}[1]{\left\{ {#1} \right\}}
\newcommand{\paren}[1]{\left( {#1} \right)}
\newcommand{\sparen}[1]{\left[ {#1} \right]}

\newcommand{\pdp}{\mathsf{PDP}}
\newcommand{\por}{\mathsf{PoR}}
\newcommand{\mspor}{\mathsf{MPoR}}

\newcommand{\succp}{\mathsf{succ}(\cP)}

\newcommand{\verifier}{\mathsf{Verifier}}

\newcommand{\enc}{\mathsf{Enc}}

\newcommand{\ext}{\mathsf{Extractor}}

\newcommand{\suc}{\mathsf{succ}}

\newtheorem{example}{Example}[section]
\usepackage{longtable}

\begin{document}

\title{Multi-prover Proof-of-Retrievability}
\author{
Maura~B.~Paterson\\
Department of Economics,
Mathematics and Statistics\\ Birkbeck, University of London,
Malet Street, London WC1E 7HX, UK \\
{\sf m.paterson@bbk.ac.uk}
\and
Douglas~R.~Stinson\thanks{D.~Stinson's research is supported by NSERC discovery grant 203114-11}
\\David R. Cheriton School of Computer Science\\ University of Waterloo,
Waterloo, Ontario, N2L 3G1, Canada \\
{\sf dstinson@math.uwaterloo.ca}
\and
Jalaj Upadhyay
\\ College of Information Science and Technology \\ Pennsylvania State University, State College, Pennsylvania 16802, U.S.A. \\
{\sf jalaj@psu.edu}
\thanks{Work done while at David R. Cheriton School of Computer Science, University of Waterloo,
Waterloo, Ontario, N2L 3G1, Canada}

}
\maketitle

\begin{abstract}
There has been considerable recent interest in ``cloud storage'' wherein a user asks a server to store a large file. One issue is whether the user can verify that the server is actually storing the file, and typically a challenge-response protocol is employed to convince the user that the file is indeed being stored correctly. The security of these schemes is phrased in terms of an extractor which will recover  the file given any ``proving algorithm'' that has a sufficiently high success probability. This forms the basis of {\em proof-of-retrievability} ($\por$)  systems. 

In this paper, we study multiple server $\por$ systems. Our contribution in multiple-server $\por$ systems is as follows.
\begin{enumerate}
	\item We formalize security definitions for two possible scenarios: (i) when a threshold  of servers succeed with high enough probability (worst-case) and (ii) when the average of the success probability of all the servers is above a threshold (average-case). We also motivate the study of confidentiality of the outsourced message. 
	\item  We give $\mspor$ schemes which are secure under both these security definitions and provide reasonable confidentiality guarantees even when there is no restriction on the computational power of the servers.	We also show how classical statistical techniques used by Paterson, Stinson and Upadhyay (Journal of Mathematical Cryptology: 7(3)) can be extended to evaluate whether the responses of the provers are accurate enough to permit successful extraction.
	\item We also look at one  specific instantiation of our construction when instantiated with the unconditionally secure version of the Shacham-Waters scheme  (Asiacrypt, 2008). This scheme gives reasonable security and privacy guarantee. We show that, in the multi-server setting with computationally unbounded provers, one can overcome the limitation that the verifier needs to store as much secret information as the provers. 
\end{enumerate}

\end{abstract}

\section{Introduction} \label{sec:introduction}
In the recent past, there has been a lot of activity on remote storage and the associated cryptographic problem of integrity of the stored data. This question becomes even more important when there are reasons to believe that the remote servers might act maliciously, i.e., one or more servers can delete (whether accidentally or on purpose) a part of the data since there is a good chance that the data will never be accessed and, hence, the client would never find out! In order to assuage such concerns, one would prefer to have a simple auditing system that convinces the client if and only if the server has the data. Such audit protocols, called {\em proof-of-retrievability} ($\por$) systems, were introduced by Juels and Kaliski~\cite{JK07}, and closely related {\em proof-of-data-possession} ($\pdp$) systems were introduced by Ateniese {\it et al.}~\cite{ABCHKPS07}.

In a $\por$ protocol, a client stores a message $m$ on a remote server and keeps only a  short private {\em fingerprint} locally. At some later  time, when the client wishes to verify the integrity of its message, it can run an audit protocol in which it acts as a verifier while the server proves that it has the client's data. The formal security of a $\por$ protocol is expressed in terms of an {\em extractor} -- there exists an extractor with (black-box or non-black box) access to the proving algorithm used by the server to respond to the client's challenge, such that the extractor retrieves the original message given any adversarial server which passes the audits  with a threshold probability. Apart from this security requirement, two practical requirements of any $\por$ system would be to have a reasonable bound on the communication cost of every audit and small storage overhead on both the client and server.  

$\por$ systems were originally defined for the single-server setting. 
However, in the real world, it is highly likely that a client would store its data on more than one  server. This might be due to a variety of reasons. For example, a client might wish to have a certain degree of redundancy if one or more servers fails. In this case, the client is more likely to store multiple copies of the same data. Another possible scenario could be that the client does not trust a single server with all of its  data. In this case, the client might distribute the data across multiple servers. Both of these settings have been studied previously in the literature.


The first such study was initiated by Curtmola {\it et al.}~\cite{CKBA08}, who considered the first of the above two cases. They addressed the problem of storing \ copies of a single file on multiple servers. This is an attractive solution considering the fact that replication is a fundamental principle in ensuring the availability and durability of data. Their system allows the client to audit a subset of servers even if some of them collude. 

On the other hand, Bowers, Juels, and Oprea~\cite{BJO09} considered the second of the above two cases. They studied a system where the client's data is distributed and stored on different servers. This  ensures that none of the servers has the whole data. 

Both of these systems covered one specific instance of the wide spectrum of possibilities when  more than one servers is involved. For example, none of the works mentioned above addresses the question of the privacy of  data. Both of them argue that, for privacy, the client can encrypt its file before storing it on the servers. These systems are secure only in the computational setting and the privacy guarantee is  dependent on the underlying encryption scheme. On the other hand, there are known primitives in the setting of distributed systems, like secret sharing schemes, that are known to be unconditionally secure. Moreover, we can also utilize  cross-server redundancy to get more practical systems. 
 

\subsection{Our Contributions}
In~\secref{model}, we give the formal description of multi-server $\por$ ($\mspor$) systems. We state the definitions for  worst-case and the average-case secure $\mspor$ systems. We also motivate the privacy requirement and state the privacy definition for $\mspor$ systems. In~\secref{prelim}, we define various primitives to the level required to understand this paper.

In~\secref{worst}, we give a construction of an $\mspor$ scheme that achieves worst-case security when the malicious servers are computationally unbounded. Our construction is based on ramp schemes and a single-server $\por$ scheme. Our  construction achieves confidentiality of the message. To exemplify our scheme, we instantiate this scheme with a specific form of ramp scheme.

In~\secref{average}, we give a construction of an $\mspor$ scheme that achieves average-case security against computationally unbounded adversaries. For an $\mspor$ system that affords average-case security, we also show that an extension of classical statistical techniques used by Paterson, Stinson and Upadhyay~\cite{PSU13} can be used to provide a basis for estimating whether the responses of the servers are accurate enough to allow successful extraction.

One of the benefits of an $\mspor$ system is that it provides cross-server redundancy. In the past, this feature has been used by Bowers {\it et al.}~\cite{BJO09} to propose a multi-server system called HAIL. We first note that the constructions in~\secref{worst} and~\secref{average} do not provide any improvement on the storage overhead of the server or the client. In~\secref{optimization}, we give a construction based on the Shacham-Waters protocol~\cite{SW08} that allows significant reduction of the storage overhead of the client in the multi-server setting. 

\subsection{Related Works}
The concept of {\it proof-of-retrievability} is due to Juels and Kaliski~\cite{JK07}. A {\sf PoR} scheme incorporates a challenge-response protocol in which a verifier can check that a message is being stored correctly, along with an {\it extractor} that will actually reconstruct the message, given the algorithm of a ``prover'' who is able to correctly respond to a sufficiently high percentage of challenges.

There are also papers that describe the closely related (but slightly weaker) idea of a {\it proof-of-data-possession scheme} (\textsf{PDP} scheme), e.g., \cite{ABCHKPS07}. A {\sf PDP} scheme permits the possibility that 
not all of the message blocks can be reconstructed. Ateniese {\it et al.}~\cite{ABCHKPS07} also introduced the idea of using {\it homomorphic authenticators} to reduce the communication complexity of the system. This scheme was improved in a follow-up work by Ateniese {\it et al.}~\cite{ADMT08}. 
Shacham and Waters~\cite{SW08} later showed that the scheme of Ateniese {\it et al.}\ \cite{ABCHKKPS11} can be transformed into a $\por$  scheme by constructing an extractor that extracts the file  from the responses of the $\prover$ on the audits. 
 
 Bowers, Juels, and Oprea~\cite{BJO09} extended the idea of Juels and Kaliski~\cite{JK07} and used error-correcting codes. The main difference in their construction is that they use the idea of an ``outer'' and an ``inner'' code (in  the same vein as concatenated codes), to get a good balance between the extra storage overhead and computational overhead in responding to the audits. 
Dodis, Vadhan, and Wichs \cite{DVW09} provided the first example of an unconditionally secure {$\por$} scheme, also constructed from an error-correcting code, with extraction performed through {\em list decoding} in conjunction with the use of an {\em almost-universal hash function}. They also give different constructions depending on the computational capabilities of the server. Paterson {\it et al.}~\cite{PSU13} studied $\por$ schemes in the setting of unconditional security, and showed some close connections to error-correcting codes.

There have been some other works that provide proof-of-storage. Ateniese {\it et al.}~\cite{AKK09} used {\em homomorphic identification schemes} to give  efficient proof-of-storage systems. Wang {\it et al.}~\cite{wang} gave the first privacy preserving public auditable proof-of-storage systems. We refer the readers to the survey by Kamara and Lauter~\cite{KL10} regarding the architecture of proof-of-storage systems. 


\paragraph{Distributed Cloud Computing.}
Proof-of-storage systems have been also studied in the setting where there is more than one server or more than one client. The first such setting was studied by Curtmola {\it et al.}~\cite{CKBA08}. They studied a multiple replica $\pdp$ system, which is the natural generalization of single server $\pdp$ system to $t$ servers. 

Bowers {\it et al.}~\cite{BJO09} introduced  a distributed system that they called {\sf HAIL}. Their system allows a set of provers to prove the integrity of a file stored by a client. The idea in {\sf HAIL} is to exploit the cross-prover redundancy. They considered an active and mobile adversary that can corrupt the whole set of provers. 

Recently, Ateniese {\it et al.}~\cite{ADDV12} considered the problem from the client side, where $n$ clients store their respective files on a single prover in a manner such that the verification of the integrity of a single client's file simultaneously gives  the integrity guarantee of the files of all the participating clients. They called such a system an {\em entangled cloud storage.}



\subsection{Comparison with Bowers, Juels, and Oprea}
The scheme of Bowers, Juels, and Oprea~\cite{BJO09} is closest to our work; however, there are some key differences. We enumerate some of them below:
\begin{enumerate}
  \item {The construction of Bowers, Juels, and Oprea~\cite{BJO09} is secure only in the computational setting, while we provide security in the setting of unconditional security.} 
   \item Bowers, Juels, and Oprea~\cite{BJO09} use various tools and algorithms to construct their systems, including error-correcting codes, pseudo-random functions, message authentication codes, and universal hash function families. On the other hand, we only use ramp schemes in our constructions, making our schemes easier to state and analyze.
   \item We consider two types of security guarantees -- the worst-case scenario and the average-case scenario (Bowers, Juels, and Oprea~\cite{BJO09} only consider worst-case scenario).
   \item The construction of Bowers, Juels, and Oprea~\cite{BJO09} only aims to protect  the integrity of the message, while we consider both the privacy and integrity of the message.
   \item We work under a stronger requirement than~\cite{BJO09} -- we require  extraction to succeed with probability equal to 1, whereas in \cite{BJO09}, extraction succeeds with probability close to 1, depending in part on properties of a certain class of hash functions used in the protocol.
\end{enumerate}

We use the term $\prover$ to identify any server that stores the file of a client. We use the term $\verifier$ for any entity that verifies whether the file of a client is stored properly or not by the server. We also assume that a file is composed of message blocks of an appropriate fixed length. If the file consists of single block, we simply call it the file. 

\section{Security Model of Multi-server $\por$ systems} \label{sec:model}
The essential components of  multi-server-$\por$ ($\mspor$) systems are natural generalizations of  single-server $\por$ systems. The first difference is that there are $\rho$ provers and the $\verifier$ might store different messages on each of them. Also, during an audit phase, the $\verifier$ can pick a subset of provers on which it runs the audits. The last crucial difference is that the $\ext$ has (black-box or non-black-box) access to a subset of proving algorithms corresponding to the provers that the $\verifier$ picked to audit. We detail them below for the sake of completeness.

Let $\prover_1, \ldots, \prover_\rho$ be a set of $\rho$ provers and let $\verifier$ be the verifier. 
The $\verifier$ has a message $m \in \cM$ from the message space $\cM$ which he redundantly encodes to $M_1, \ldots, M_\rho$. 
\begin{enumerate}
	\item In the keyed setting, the $\verifier$ picks $\rho$ different keys $(K_1, \ldots, K_\rho)$, one for each of the corresponding provers.
	\item The $\verifier$ gives $M_i$ to $\prover_i$. In the case of a keyed scheme,  $\prover_i$ may be also given an additional tag $S_i$ generated using the key, $K_i$, and $M_i$.
	\item  The $\verifier$ stores some sort of information (say a {\em fingerprint} of the encoded message) which allows him to verify the responses made by the provers. 
	\item On receiving the encoded message $M_i$, $\prover_i$ generates a proving algorithm $\cP_i$, which it uses to generate its responses during the auditing phase. 
	\item At any time, $\verifier$
picks an index $i$, where $1\leq i \leq \ell$, and engages in a challenge-response
protocol with $\prover_i$. In one  execution of challenge-response protocol, $\verifier$ picks a challenge $c$ and gives it to $\mathsf{Prover}_i$, and the prover responds with $\varrho$. The $\verifier$ then verifies the correctness of the response (based on its fingerprint).
	\item The success probability $\suc(\cP_i)$ is the probability, computed over all the challenges, with which the $\verifier$ accepts the response sent by $\prover_i$. 
	\item The $\ext$ is given a subset $S$ of the proving algorithms $\cP_1, \ldots, \cP_\rho$ (and in the case of a keyed scheme, the corresponding subset of the keys, $\{K_i : i \in S\}$), and outputs a message $\widehat{m}$. The $\ext$ succeeds if $\widehat{m}=m$.
\end{enumerate}

The above framework does not restrict any provers from interacting with other provers when they receive the encoded message. However, we assume that they do not interact {\em after} they have generated a proving algorithm. If we do not include this restriction, then it is not hard to see that  one cannot have any meaningful protocol. For example, if provers can interact after they receive the encoded message, then it is possible that one prover stores the entire message and the other provers just relay the challenges to this specific prover and relay back its response to the verifier. 

In contrast to a single-prover $\por$ scheme, there are two possible ways in which one can define the security of a multiple prover $\por$ system. We define them next.

The first security definition corresponds to the ``worst case" scenario and is the natural generalization of a single-server $\por$ system.
\begin{definition} \label{defn:worst}
	A $\rho$-prover $\mspor$ scheme is {\em $(\eta,\nu,\tau,\rho )$-threshold secure} if there is an $\ext$ which, when given any $\tau$ proving algorithms, say $\cP_{i_1}, \ldots \cP_{i_\tau}$, succeeds with probability at least $\nu$ whenever 
	\[  \suc(\cP_j) \geq \eta \qquad  \text{for all}~j \in I, \] where $I =\set{i_1, \ldots, i_\tau}$.
\end{definition}

We note that when $\rho =\tau=1$, we get a standard single-server $\por$ system. Moreover, the definition captures the worst-case scenario in the sense that  it only guarantees extraction if there exists a set of $\tau$ proving algorithms, all of which  succeed with high enough probability. 

The above definition requires that all the $\tau$ servers succeed with high enough probability. On the other hand, it might not be the case that all the proving algorithms of   the servers picked by the $\verifier$  succeed with the required probability. In fact, even verifying whether or not all the $\tau$ proving algorithms have high enough success probability to allow successful extraction might  be difficult (see, for example~\cite{PSU13} for more details about this). However, it is possible that some of the proving algorithms succeed with high enough probability to compensate for the failure of the rest of the proving algorithms. For instance, since the provers are allowed to interact before they specify their proving algorithms, it might be the case that the colluding provers decide to store most of the message on a single prover. In this case, even a weaker guarantee that the average success probability is high enough might be sufficient to guarantee a successful extraction. In other words, it is possible to state (and as we show in this paper, achieve) a security guarantee with respect to the average case success probability over all the proving algorithms.
\begin{definition} \label{defn:average}
A $\rho$-prover $\mspor$ scheme is {\em $(\eta,\nu,\rho )$-average secure} if the $\ext$ succeeds with probability at least $\nu$ whenever 
\[ \frac{1}{\rho}\sum_{i=1}^\rho \suc(\cP_i) \geq  \eta .\]
\end{definition}
Note that the average-case secure system reduces to the standard $\por$ scheme (with $\tau=\rho$) when $\rho =1$. The following example illustrates that average-case security is possible even when an $\mspor$ system is not possible as per Definition~\ref{defn:worst}.

\begin{example} \label{eq:average}
Suppose $\eta = 0.7, \nu = 0$ and $\rho = 3$.
Further, suppose that $\suc (\cP_1) = 0.9,~\suc(\cP_2) = 0.6$ and $\suc(\cP_3) = 0.6$.
Then the hypotheses of Definition~\ref{defn:worst} are
not satisfied for $\tau = 2$. So even if the $\mspor$
scheme is $(\eta,\nu,\tau,\rho)$-threshold secure,
we cannot conclude that the Extractor will succeed.
On the other hand, for the assumed success probabilities,
the hypotheses of Definition~\ref{defn:average} are satisfied.
Therefore, if the $\mspor$ scheme is 
$(0.7,\nu,\tau)$-average secure, 
the Extractor will succeed.
\end{example}


\paragraph{\scshape Privacy Guarantee.}
We mentioned at the start of this section that $\por$ systems were introduced and studied to give  assurance of the integrity of the data stored on remote storage. However, the confidentially aspects of data  have not been studied formally in the area of cloud-based $\por$ systems. There have been couple of {\it ad hoc} solutions that have been proposed in which the messages are encrypted and then stored on the cloud~\cite{CKBA08}. We believe that, in addition to the standard integrity requirement,  privacy of the stored data when multiple provers are involved is also an important requirement. We model the privacy requirement as follows:

\begin{definition} \label{defn:privacy}
	An $\mspor$ system is called $t$-private if no set $\cA$ of adversarial provers  of size at most $t$ learns anything about the message stored by the $\verifier$.
\end{definition}

Note that $t=0$ corresponds to the case when the $\mspor$ system does not provide any confidentiality to the message. The above definition captures the idea that, even if $t$ provers collude, they do not learn anything about the message. We remark that we can achieve confidentiality without encrypting the message by using secret sharing techniques.

\paragraph{Notation.}We fix the letter $m$ for the original message, $\mathcal{M}$ to denote the space from which the message $m$ is picked, and $M$ to denote the encoded message. 
We fix $\nu$ to denote the failure probability of the extractor and $\eta$  to denote the success probability of a proving algorithm. In this paper, we are mainly interested in the case when $\nu=0$ for both the worst-case and the average-case security. 
 We use $n$ to denote the number of  message blocks, assuming the underlying $\por$ system breaks the message into blocks. 

\section{Primitives Used in This Paper} \label{sec:prelim}

\subsection{Ramp Schemes} \label{sec:secret}
In our construction, we use a primitive related to secret sharing schemes known as {\em ramp schemes}. A  {\em secret sharing scheme} allows a trusted dealer to share a secret between $n$ players so that certain subsets of players can reconstruct the secret from the shares they hold~\cite{Blakley79,Shamir79}. 

It is well-known that the size of each player's share in a secret sharing scheme must be at least the size of the secret. If the secret that is to be shared is large, then this constraint can be very restrictive. Schemes for which we can get a certain form of trade-off between share size and security are known as {\em ramp schemes}~\cite{BM85}.

\begin{definition}\label{defn:rampscheme}
{\em (Ramp Scheme).} Let $\tau_1,\tau_2$, and $n$ be positive integers such that $\tau_1 < \tau_2 \leq n$. A {\em $(\tau_1, \tau_2, n)$-ramp scheme} is a pair of algorithms: $(\mathsf{ShareGen, Reconstruct})$ such that, on input a secret $\mathsf{S}$, $\mathsf{ShareGen}(\mathsf{\mathsf{S}})$ generates $n$ shares, one for each of the $n$ players, such that the following two properties hold: (i) Reconstruction: any subset of $\tau_2$ or more players can pool together their shares and use $\mathsf{Reconstruct}$ to compute the secret $\mathsf{S}$ from the shares that they collectively hold, and (ii) Secrecy: no subset of $\tau_1$ or fewer players can determine any information about the ${secret}~\mathsf{S} $.
\end{definition}

\begin{example} Suppose the dealer wishes to set up $(2,4,n)$-ramp scheme with the secret $(a_0,a_1)$. The dealer picks a finite field $\F_q$ with  $q >n$ such that $a_0,a_1 \in \F_q$. The dealer picks random elements $a_2,a_3$  independently from the field $\F_q$ and construct the following polynomial of degree $3$ over the finite field $\F_q$:  $f(x)=a_0 +a_1x+ a_2x^2 +a_3x^3$. The share for any player $\cP_i$ is generated by computing $s_i=f(i)$. It is easy to see that if two or fewer players come together, they do not learn any information about the secret, and if at least four players come together, they can use Lagrange's interpolation formula to compute the function $f$ as well as the secret. However, if  three players  pool together their shares, then they can learn some partial information about one of the other player's share. For concreteness, let $q ={17}$. Then $5a_1 \equiv 7s_3+ 9s_6 + s_{15} \mod 17$; therefore, players $\cP_3,\cP_6$, and $\cP_{15}$ can compute the value of $a_1$. 
\end{example}

For completeness, we review some of the basic theory concerning the construction of ramp schemes.
Linear codes have been used to construct ramp-schemes  for over thirty years since the work of McEliece and Sarwate~\cite{MS81}. We will consider a construction from an arbitrary code in  this paper. The following relation between an arbitrary code (linear or non-linear) and a ramp scheme is  was shown by Paterson and Stinson~\cite{PS13}.
\begin{theorem} \label{thm:PS13}
Let $\mathcal{C}$ be a code of length $N$, distance $\mathsf{d}$ and dual distance $\mathsf{d}^\bot$. Let $1 \leq \mathsf{s} <\mathsf{d}^\bot-2$. Then there is a $(\tau_1, \tau_2, {N}-\mathsf{s})$ ramp scheme 
, where $\tau_1=\mathsf{d}^\bot-\mathsf{s}-1$ and $\tau_2={N}-\mathsf{d}+1$.
\end{theorem}
Here $\mathsf{s}$ is the {\em rate} of the ramp scheme. If $\mathbf{G}$ is a generator matrix of a code $C$ of dimension $k$, then $|C| = q^k \geq q^{\mathsf{d}^\bot-1}$. In other words, $k \geq {\mathsf d}^\bot-1.$

\begin{construction} \label{construction:ramp}
The construction of a ramp scheme from a code is as follows. Let $s$ and $\rho$ be  positive integers and let $(m_1, \ldots, m_{\mathsf{s}}) \in \F^{\mathsf{s}}$ be the message. Let $C$ be a code of length $n=\rho+{s}$ defined over a finite field $\F$. We also require that the first $\mathsf{s}$ entries of a codewords is the message to be encoded, i.e., the corresponding generator matrix is in  standard form. Select a random codeword $(\mathbf{c}_1=m_1, \ldots, \mathbf{c}_\mathsf{s}=m_\mathsf{s}, \mathbf{c}_{\mathsf{s}+1}, \ldots, \mathbf{c}_{\rho+\mathsf{s}}) \in C $, and define the shares as $(\mathbf{c}_{\mathsf{s}+1}, \ldots, \mathbf{c}_{\rho+\mathsf{s}})$. 
\end{construction}

\begin{example} \label{eg:rscode}
One can use a Reed-Solomon code to construct a ramp scheme~\cite{MS81}. Let $q$ be a prime and $1 \leq \mathsf{s} < t\leq n <q $. It is well known that, for a prime $q$,  there is an $[N,k,N-\tau+1]_q$ Reed-Solomon code with $\mathsf{d}^\bot=\tau+1$. This implies a $(\tau-\mathsf{s},\tau,N)$-ramp scheme over $\F_q$. 
\end{example}



\subsection{Single-prover $\por$-system}
We start by fixing some notation for {\sf PoR} schemes that we use throughout the paper. Let $\Gamma$ be the {\em challenge space} and $\Delta$ be the {\em response space}. We denote by $\gamma= |\Gamma |$ the size of a challenge space. Let $\cM^*$ be the space of all encoded messages. The {\em response function} $\rho : \mathcal{M}^* \times \Gamma \rightarrow \Delta $
computes the response $r = \rho(M,c)$ given the encoded message $M$ and the challenge $c$.

For an encoded message $M \in \mathcal{M}^*$, we define the {\em response vector} $r^M$ that contains all the  responses
to all possible challenges for the encoded message $M$. Finally, define the {\em response code}
 of the
scheme to be
\[ {\cR} = \{ r^M : M \in {\cM}^* \} .\]
The codewords in ${\cR}$ are just the response vectors that we defined above.
Paterson {\it et al.}~\cite{PSU13} proved the following result for a single-prover {\sf PoR} scheme.
\begin{theorem}   \label{thm:single}
	Suppose that $\cP$ is a proving algorithm for a {\sf PoR} scheme with response code $\cR$. If the success probability of the corresponding proving algorithm satisfies ${\sf succ}(\cP) \geq 1 - \widetilde{\mathsf{d}}/2\gamma$, where $\widetilde{\mathsf{d}}$ is the Hamming distance of the code $\cR$ and $\gamma$ is the size of the challenge space, then the extractor described in~\figref{extractor} always outputs the message $m$.
\end{theorem}

\begin{figure}[t]
\begin{center}
\fbox{
    \begin{minipage}{5.5in}
    \begin{enumerate}
\item On input $\cP$, compute the vector $R' = (r'_c : c \in \Gamma)$,
where $r'_c = \cP(c)$ for all $c \in \Gamma$ (i.e., for every $c$, $r'_c$ is
the response computed by $\cP$ when it is given
the challenge $c$).
\item Find $\widehat{M} \in \mathcal{M}^*$ so that $\mathsf{dist} (R', r^{\widehat{M}})$ is minimised.
\item Output $\widehat{m} = e^{-1}(\widehat{M})$.
\end{enumerate}
\end{minipage}
  }
\caption{$\ext$  for~\thmref{single}}\label{fig:extractor}
\end{center}
\end{figure}

If we cast this in the security model defined in~\secref{introduction} (Definition~\ref{defn:worst} and Definition~\ref{defn:average}), then we have the following theorem.
\begin{theorem} \label{thm:PSU13} 
	Suppose that $\cP$ is a proving algorithm for a single server $\por$ scheme with response code $\cR$. Then there exists a $(1 - \widetilde{\mathsf{d}}/2\gamma, 0,1,1)\mbox{-}\mspor$ system, where $\widetilde{\mathsf{d}}$  is the Hamming distance of the code $\cR$ and $\gamma$ is the size of the challenge space $\Gamma$. 
\end{theorem}

Paterson {\it et al.}~\cite{PSU13} gave a modified version of the Shacham-Waters scheme which they showed is secure in the unconditional security setting.  They argued that, in the setting of  unconditionally security, any keyed $\por$ scheme should be considered to be secure when the success probability of the proving algorithm $\cP$, denoted by $\succp$, is defined as the average success probability of the prover  over all possible keys (\thmref{SW08}).  The same reasoning extends to  $\mspor$ systems. Therefore, in what follows and in~\secref{optimization}, when we say a scheme is an $(\eta,\nu,\tau,\rho)$-threshold secure scheme, the term $\eta$ is the average success probability where the average is computed over all possible keys.  We denote the average success probability of a prover $\cP$ over all possible keys by $\suc_\mathsf{avg}(\cP)$. 
Paterson {\it et al.}~\cite{PSU13} showed the following:
\begin{theorem}
\label{SWsucc.thm} \label{thm:SW08}
Let $\F_q$ be the underlying field and let $\ell \geq 1$ be the hamming weight of the challenges made by the $\verifier$. Let $\mathsf{d}$ be the hamming distance of the space of encoded message, $\mathcal{M}^*$.  Suppose that
\begin{equation}
\label{eq2.6} \mathsf{succ}_{\mathsf{avg}}(\mathcal{P}) \gtrsim 1 - \frac{\mathsf{d}^*(q-1)}{2 \gamma q},
\end{equation}
where $\gamma=q^n$ is the size of the challenge space and $\mathsf{d}^*$ is given by 
\begin{align}
{\mathsf{d}^*} &\approx \binom{n}{\ell} (q-1)^{\ell} - \binom{n-\mathsf{d}}{\ell} (q-1)^{\ell} -
\sum_{w \geq 1} \binom{\mathsf{d}}{w} \binom{n-\mathsf{d}}{\ell - w} \frac{(q-1)^{\ell}}{q}. \label{eq:d}
\end{align}

 Then there exists an $\ext$ that always output $\widehat{m} = m$.
\end{theorem}

\section{Worst-case $\mspor$ Based on Ramp Scheme} \label{sec:worst}
In this section, we give our first construction that achieves a worst-case security guarantee. The idea is to use a $(\tau_1, \tau_2, \rho )$-ramp scheme in conjunction with a single-server-$\por$ system. The intuition behind the construction is that the underlying $\por$ system along with the ramp scheme provides the retrievability guarantee and the ramp scheme provides the confidentiality guarantee. 

We first present a schematic diagram of the working of an $\mspor$ in~\figref{msporscheme} and illustrate the scheme with the help of following example. We provide the details of the construction in~\figref{ramp}.

\begin{example}
	Let $\rho=6$. Suppose the $\verifier$ and the provers use a $\por$ system $\Pi$. Let the message to be stored be $(15,3)$. The $\verifier$ picks $q=17$ and chooses two random elements $1,2 \in \F_{17}$ to construct a polynomial $f(x)=15 + 3 x + x^2 + 2x^3$. 
	The $\verifier$ picks an encoding function $e(\cdot)$ and stores $e(4)$ on $\prover_1$, $e(7)$ on $\prover_2$,  $e(2)$ on $\prover_3$,   $e(1)$ on $\prover_4$,    $e(16)$ on $\prover_5$, and      $e(8)$ on $\prover_6$. 
     
     Let us suppose that the $\por$ scheme is  such that, for  a random challenge vector of dimension $\rho$, say $\begin{pmatrix} 5, & 2, & 9, & 13, & 5, & 6\end{pmatrix}$, where the $i$-th entry would be a challenge to $\prover_i$, the corresponding responses of the provers form a vector $\begin{pmatrix} 3, & 14, & 1, & 13, & 12, & 14  \end{pmatrix}$, where the $\text{\sc Resp}_i$ is the correct response of the $\prover_i$. In other words, on challenge $5$ to $\prover_1$, the correct response is $3$, and so on. 
     
     During the audit phase,  the $\verifier$ picks any four provers and sends the  challenges to the provers. Once all the provers that he chose reply, he verifies their response. For example, suppose the $\verifier$ picks $\prover_1,~\prover_3,~\prover_4,$ and $\prover_6$. The $\verifier$ then  sends the challenge $5$ to $\prover_1$, $9$ to $\prover_3$, $13$ to $\prover_4$, and $6$ to $\prover_6$. If it gets the response $3,~1,~13,$ and $14$ back, it accepts; otherwise, it rejects.
\end{example}

\usetikzlibrary{arrows,shapes,positioning,shadows,trees}
\tikzstyle{l} = [draw, -latex',thick]
\tikzstyle{myarrows}=[line width=1mm,draw=blue,-triangle 45,postaction={draw, line width=2mm, shorten >=3mm, -}]

\tikzset{
  basic/.style  = {draw, text width=2cm, drop shadow, font=\sffamily, rectangle,fill=green!30},
  root/.style   = {basic, text width=11cm, rounded corners=2pt, thin, align=center, fill=green!30},
  level 2/.style = {basic, rounded corners=6pt, thin,align=center, fill=green!60,
                   text width=8em,node distance=3cm},
  level 3/.style = {basic, thin, align=center, fill=red!30, text width=6.5em,node distance=3cm}
}
\begin{figure}
\begin{center}
\begin{tikzpicture}[node distance=2.8cm, auto,
  level 1/.style={sibling distance=6cm}, {node distance=4 cm}
    emph/.style={edge from parent/.style={red,very thick,draw}},
  >=latex]

\node[root] {Message in the form of $\mathsf{s}$ bits}
  child {node[level 2] (c1) {Share $1$ of $\mathsf{Ramp}$ scheme}}
  child {node[level 2] (c2) {Share $i$ of $\mathsf{Ramp}$ scheme}}
  child {node[level 2] (c3) {Share $\rho$ of $\mathsf{Ramp}$ scheme}};

\begin{scope}[every node/.style={level 3}]
\node [below of = c1] (c11) {Block stored on $\prover_1$};
\node [below of = c2] (c21) {Block stored on  $\prover_i$};
\node [below of = c3] (c31) {Block stored on $\prover_\rho$};
\end{scope}

\draw[myarrows](c1) to node {$\Pi$}(c11);
\draw[myarrows](c2) to node {$\Pi$}(c21);
\draw[myarrows](c3) to node {$\Pi$}(c31);

\path (c1) -- (c2) node [midway] {$\ldots$};
\path (c2) -- (c3) node [midway] {$\ldots$};
\path (c11) -- (c21) node [midway] {$\ldots$};
\path (c21) -- (c31) node [midway] {$\ldots$};


\end{tikzpicture}
\caption{Schematic View of $\mathsf{Ramp}\mbox{-}\mspor$ System} \label{fig:msporscheme}
\end{center}
\end{figure}

We  note  one of the possible practical deployments of the $\mathsf{Ramp}\mbox{-}\mspor$ stated in~\figref{ramp}. Let $m$ be a message that consists of $sk$ elements from $\F_q$. The $\verifier$ breaks the message into $k$ blocks of length $s$ each. It then invokes a $(\tau_1,\tau_2,n)\mbox{-}\mathsf{Ramp}$ scheme on each of these blocks to generate $n$ shares of each of the $k$ blocks. The $\verifier$ then runs a $\por$ scheme $\Pi$ to compute the encoded message to be stored on each of the servers by encoding its $k$ shares, one corresponding to each of the $k$ blocks.

\begin{figure} [t]
\begin{center}
\fbox{
\begin{minipage}[l]{5.5in}
\medskip
{
{\bf Input:} The  $\verifier$ gets the message $m$ as input. Let $\prover_1, \ldots, \prover_\rho$ be the set of $\rho$ provers. 
\begin{description}
	\item [Initialization Stage.] The $\verifier$ performs the following steps for storing the message
	\begin{enumerate}
		\item The $\verifier$ chooses a single-server $\por$ system $\Pi$ and a $(\tau_1,\tau_2,\rho )$-ramp scheme $\mathsf{Ramp}=(\mathsf{ShareGen, Reconstruct})$. \label{step:chose}
		\item The $\verifier$ computes $\rho$ shares of the message using the ramp scheme $(m_1, \ldots, m_\rho) \gets \mathsf{ShareGen}(m)$.
		\item The $\verifier$ runs $\rho$ independent copies of $\Pi$ and generates the encoded share $M_i=e(m_i) \in \mathcal{M}$ corresponding to each $1 \leq i \leq \rho$. 
		\item The $\verifier$ stores $M_i$ on $\prover_i$.
	\end{enumerate}
	\item [Challenge Phase:] During the audit phase, $\verifier$ picks a prover, $\prover_i$, and runs the challenge-response protocol of $\Pi$ with $\prover_i$.
 \end{description}
}
\end{minipage}
}\caption{ Worst-case Secure $\mspor$ Using a Ramp-scheme ($\mathsf{Ramp}\mbox{-} \mspor$).} \label{fig:ramp} 
\end{center}
\end{figure}

We prove the following security result for the $\mspor$ scheme presented in~\figref{ramp}.
\begin{theorem} \label{thm:ramp}
	Let $\Pi$ be an $(\eta,0,1,1)\mbox{-}$ threshold-secure $\mspor$ with a response code of Hamming distance $\widetilde{\mathsf{\mathsf{d}}}$ and the size of challenge space $\gamma$. Let $\mathsf{Ramp}=(\mathsf{ShareGen, Reconstruct})$ be a $(\tau_1,\tau_2,\rho )$-ramp scheme. Then  $\mathsf{Ramp\mbox{-}}\mspor$, defined in~\figref{ramp}, is an $\mspor$ system with the following properties:
	\begin{enumerate}
		\item Privacy: $\mathsf{Ramp\mbox{-}}\mspor$ is $\tau_1$-private.
		\item Security: $\mathsf{Ramp\mbox{-}}\mspor$ is $(\eta,0,\tau_2,\rho )$-threshold secure, where $\eta=1-\widetilde{\mathsf{d}}/2\gamma$.
	\end{enumerate}
\end{theorem}
\begin{proof}
	 The privacy guarantee of $\mathsf{Ramp\mbox{-}}\mspor$ is straightforward from the privacy property of the underlying ramp scheme. 
	 
	 For the security guarantee, we need to demonstrate an $\ext$ that outputs a message $\widehat{m} =m$ if at least $t$ servers succeed with probability at least $\eta=1-\widetilde{\mathsf{d}}/2 \gamma$. The description of our $\ext$ is as follows:
	\begin{enumerate}
		\item $\ext$ chooses $\tau_2$ provers and runs the extraction algorithm of the underlying single-server $\por$ system on each of these provers. In the end, it outputs $\widehat{M}_{i_j}$ for the corresponding provers $\prover_{i_j}$. It defines $\mathcal{S} \gets \{\widehat{M}_{i_1}, \ldots, \widehat{M}_{i_{\tau_2}} \}$.
		\item $\ext$ invokes the $\mathsf{Reconstruct}$ algorithm of the underlying ramp scheme with the  elements of $\mathcal{S}$. It outputs whatever $\mathsf{Reconstruct}$ outputs.
	\end{enumerate}	  
	 Now note that the $\verifier$ interacts with every $\prover_i$ independently. We know from the security of the underlying single-server-$\por$ scheme (\thmref{single}) that there is an extractor that always outputs the encoded message whenever $\suc(\cP_i) \geq \eta.$ Therefore, if all the $\tau_2$ chosen proving algorithms succeed with probability at least $\eta$, then the set $\mathcal{S}$  will have  $\tau_2$ correct shares. From the correctness of the $\mathsf{Reconstruct}$ algorithm, we know that the message output in the end by $\ext$ will be the message $m$. 
\end{proof}

As a special case of the above, we get a simple $\mspor$ system which uses a {\em replication code}. A replication code has an encoding function $\enc: \Lambda \rightarrow \Lambda^\rho$ such that $\enc(x) = (\underbrace{x,x,\ldots, x}_{\rho \text{ times}})$ for any $x \in \Lambda$. This is the setting considered by Curtmola {\it et al.}~\cite{CKBA08}.

We call a $\mathsf{Ramp\mbox{-}}\mspor$ scheme based on a replication code a $\mathsf{Rep\mbox{-}}\mspor$. The schematic description of the scheme is presented in~\figref{repscheme} and the scheme is presented in~\figref{average}. Since a $\rho$-replication code is a $(0,1,\rho)$-ramp scheme, a simple corollary to~\thmref{ramp} is the following.
\begin{corollary} \label{cor:rep}
	Let  $\Pi$ be a $(\eta,0,1,1)\mbox{-}\mspor$ system with a response code of Hamming distance $\widetilde{\mathsf{d}}$ and the size of challenge space $\gamma$. Then  $\mathsf{Rep\mbox{-}}\mspor$, formed by instantiating $\mathsf{Ramp}\mbox{-}\mspor$ with the replication code based Ramp scheme, is a $\mspor$ system with the following properties:
	\begin{enumerate}
		\item Privacy: It is $0$-private.
		\item Security: It is $(\eta,0,1,\rho )$-threshold secure, where $\eta=1-\widetilde{\mathsf{d}}/2\gamma$.
	\end{enumerate}
\end{corollary}

The issue with $\mathsf{Rep}\mbox{-}\mspor$ scheme is that there is no confidentiality of the file. We will come back to this issue later in~\secref{extension}.

\section{Average-case Secure $\mspor$ System} \label{sec:average}
In general, it is not possible to verify with certainty whether the success probability of a proving algorithm is above a certain threshold; therefore, in that case, it is unclear how  the $\ext$ would know which proving algorithms to use for  extraction as described in~\secref{worst}. In this section, we analyze the average-case security properties of the replication code based scheme, $\mathsf{Rep}\mbox{-}\mspor$, described in the last section. This allows us an alternative guarantee that allows successful extraction where the extractor need not worry whether a certain proving algorithm succeeds with high enough probability or not.

\begin{figure}[t]
\begin{center}
\begin{tikzpicture}[node distance=4cm, auto,
  level 1/.style={sibling distance=6cm}, {node distance=4 cm}
  emph/.style={edge from parent/.style={red,very thick,draw}},
  >=latex
]

\node[root] {Message $m$}
  child {node[level 2] (c1) {M=e(m)} edge from parent node[above left] {$e$}}
  child {node[level 2] (c2) {M=e(m)} edge from parent node[left] {$e$}}
  child {node[level 2] (c3) {M=e(m)} edge from parent node[above right] {$e$}};

\begin{scope}[every node/.style={level 3}]
\node [below of = c1] (c11) {$M$ stored on $\prover_1$};
\node [below of = c2] (c21) {$M$ stored on  $\prover_i$};
\node [below of = c3] (c31) {$M$ stored on $\prover_\rho$};
\end{scope}

\draw[myarrows](c1) to  (c11);
\draw[myarrows](c2) to  (c21);
\draw[myarrows](c3) to  (c31);

\path (c1) -- (c2) node [midway] {$\ldots$};
\path (c2) -- (c3) node [midway] {$\ldots$};
\path (c11) -- (c21) node [midway] {$\ldots$};
\path (c21) -- (c31) node [midway] {$\ldots$};


\end{tikzpicture}
\caption{Schematic View of $\mathsf{Rep\mbox{-}}\mspor$} \label{fig:repscheme}
\end{center}
\end{figure}


Recall the scenario introduced in Example~\ref{eq:average}.
Here we assumed $\suc(\cP_1) = 0.9,~
\suc(\cP_2) = 0.6$ and $\suc(\cP_3) = 0.6$ for three provers.
Suppose that successful extraction for a particular prover $\cP_i$ 
requires $\suc(\cP_2)\geq 0.7$. Then extraction would work on 
only one of these three provers. On the other hand,
suppose we have an average-case secure $\mspor$ in which 
extraction is successful if the average success probability 
of the three provers is at least $0.7$. Then the success 
probabilities assumed above would be sufficient to 
guarantee successful extraction.
\begin{figure} [t]
\begin{center}
\fbox{
\begin{minipage}[l]{5.5in}
\medskip
{
{\bf Input:} The verifier $\verifier$ gets the message $m$ as input. Let $\prover_1, \ldots, \prover_\rho$ be the set of $\rho$ provers. 
\begin{description}
	\item [Initialization Stage.] The $\verifier$ performs the following steps for storing the message
	\begin{enumerate}
		\item The $\verifier$ chooses a single-server $\por$ system $\Pi$.
		\item Using the encoding scheme of $\Pi$, the $\verifier$ generates the encoded message $M=e(m) \in \mathcal{M}$ for $1 \leq i \leq n$. 
		\item The $\verifier$ stores the message $M$ on all $\prover_i$ for $1 \leq i \leq n$.
	\end{enumerate}
	\item [Challenge Phase:] During the audit phase, $\verifier$  runs the challenge-response protocol of $\Pi$ independently on each server.
 \end{description}
}
\end{minipage}
}\caption{Average case Secure $\mspor$  ($\mathsf{Rep\mbox{-}}\mspor$).} \label{fig:average} 
\end{center}
\end{figure}

\begin{theorem}
	Let  $\Pi$ be a single-server $\por$ system with a response code of Hamming distance $\widetilde{\mathsf{d}}$ and the size of challenge space $\gamma$. Then  $\mathsf{Rep\mbox{-}}\mspor$, defined in~\figref{average}, is an $\mspor$ system with the following properties:
	\begin{enumerate}
		\item Privacy: $\mathsf{Rep\mbox{-}}\mspor$ is $0$-private.
		\item Security: $\mathsf{Rep\mbox{-}}\mspor$ is $(1-\widetilde{\mathsf{d}}/2\gamma,0,\rho )$-average secure.
	\end{enumerate}
\end{theorem}
\begin{proof}
Since the message is stored in its entirety on each of the servers, there is no confidentiality. 
	 
	 For the security guarantee, we need to demonstrate an $\ext$ that outputs a message $\widehat{m} =m$ if average success probability of all the provers is at least $\eta=1-\widetilde{\mathsf{d}}/2\gamma$. The description of our $\ext$ is as follows:
	\begin{enumerate}
	\item For all $1 \leq i \leq n$, use $\cP_i$ to compute the vector $R_i = (r_c^{(i)} : c \in \Gamma)$, 
where $r_c^{(i)} = \cP_i(c)$ for all $c \in \Gamma$ (i.e., for every $c$, $r_c^{(i)}$ is
the response computed by $\cP_i$ when it is given
the challenge $c$),  
	\item Compute $R$ as a concatenation of $R_1, \ldots, R_\rho$ and  find $\widehat{M}:=\paren{\widehat{M}_1, \ldots, \widehat{M}_\rho}$ so that $\mathsf{dist} (R, r^{\widehat{M}})$ is minimized, and 
	\item Compute $m=e^{-1}(\widehat{M})$. 	
	\end{enumerate}	  
Now note that $\verifier$ interacts with each $\prover_i$ independently and $\ext$ uses the challenge-response step with independent challenges. Let $\eta_1, \ldots, \eta_\rho$ be the success probabilities of the $\rho$ proving algorithms. Let $\bar{\eta}$ be the average success probability over all the servers and challenges. Therefore, $\bar{\eta} = \rho^{-1}\sum_{i=1}^\rho \eta_i$. 


First note that, in the case of~\figref{average}, the response code is of the form \[ \set{(\underbrace{r,r,\ldots, r}_{\rho \text{ times}}): r \in \cR}. \] 

It is easy to see that the distance of the response code is $\rho  \widetilde{\mathsf{d}}$ and the length of a challenge is $\rho  \gamma$. From the definition of the extractor and~\thmref{single}, it follows that the extraction succeeds if 
\[ \frac{\eta_1 + \ldots + \eta_\rho}{\rho} = \bar{\eta} \geq 1 - \frac{\widetilde{\mathsf{d}}}{2 \gamma}. \]
\end{proof}

\subsection{Hypothesis Testing for $\mathsf{Rep}\mbox{-}\mspor$}

For the purposes of auditing whether a file is being stored appropriately, it is necessary to have a mechanism for determining whether the success probability of a prover is sufficiently high.  In the case of replication code based $\mspor$ with worst-case security, we are interested in the  success probabilities of individual provers, and the analysis can be carried out as detailed in Paterson {\it et al.}~\cite{PSU13}.  In the case of $\mathsf{Rep}\mbox{-}\mspor$, however, we wish to determine whether the {\em average} success probability of the set of provers $\{\cP_1,\cP_2,\dotsc,\cP_\rho\}$ is at least $\eta$. This amounts to distinguishing the null hypothesis
\begin{align*}
H_0: \mathsf{avg}\mbox{-}\suc(\cP_i)&<\eta \intertext{from the alternative hypothesis}
H_1:\mathsf{avg}\mbox{-}\suc(\cP_i)&\geq\eta.
\end{align*}

Suppose we send $c$ challenges to each server.  If a given server $\cP_i$ has success probability $\suc(\cP_i)$, then the number of correct responses received follows the binomial distribution ${\sf B}(c,\suc(\cP_i))$.  If the success probabilities $\suc(\cP_i)$ were the same for each server, then the sum of the number of successes over all the servers would also follow a binomial distribution.  However, we are also  interested in the case in which these success probabilities differ, in which case the total number of successes follows a {\em poisson-binomial distribution,} which is more complicated to work with.  In order to establish a test that is conceptually and computationally easier to apply, we will instead rely on the observation that, in cases where the average success probability is high enough to permit extraction,  the failure rates of the servers are relatively low.  

For a given server $\cP_i$, let $f_i=1-\suc(\cP_i)$ denote the probability of failure.  For $r$ challenges, the number of failures follows the binomial distribution ${\sf B}(c,f_i)$.  Provided that $r$ is sufficiently large and $f_i$ is sufficiently low, then ${\sf B}(c,f_i)$ can be approximated by the {\em poisson distribution} ${\sf Pois}(cf_i)$.  The poisson distribution ${\sf Pois}(\lambda)$ is used to model the scenario where discrete events are occurring independently within a given time period with an expected rate of $\lambda$ events during that period.  The probability of observing $k$ events within that period is given by $$P(k)=\frac{e^ {-\lambda}\lambda^k}{k!}.$$ The mean and the variance of ${\sf Pois}(\lambda)$ is equal to $\lambda$.  For our purposes, the advantage of using this approximation is the fact that the sum of $\rho$ independent variables following the poisson distributions ${\sf Pois}(\lambda_1),{\sf Pois}(\lambda_2),\dotsc ,{\sf Pois}(\lambda_\rho)$ is itself distributed according to the poisson distribution ${\sf Pois}(\lambda_1+\lambda_2 + \dotsc +\lambda_\rho)$, even when the $\lambda_i$ all differ.  In the case where the average failure probability is low, the distribution ${\sf Pois}(c(f_1+f_2+\dotsc + f_\rho))$ should provide a reasonable approximation to the actual distribution of the total number of failed challenges.


\begin{example}
To demonstrate the appropriateness of the Poisson approximation for this application, suppose we have five servers, whose failure probabilities are expressed as $\mathbf{f}=(f_1,f_2,\dotsc,f_{5})$.  Let $t$ be the number of trials per server and $b$ the total number of observed failures out of the $5 t$ trials.  Table~\ref{table:poision} give both the exact cumulative probability $\p [B\leq b]$ of observing up to $b$ failures, and the Poisson approximation $\p_{\sf Pois}[B\leq b]$ of this cumulative probability, for a range of values for $\mathbf{f}$.  

\begin{longtable} {|cccc|}
\caption[Effect of Approximation by Poisson Distribution]{Comparison Between Exact Cumulative Probability and Approximation by Poisson Distribution}  \label{table:poision} \\

\multicolumn{4}{c}{$\mathbf{f}=(0.1,0.1,0.1,0.1,0.1)$}\\
\hline

$t$ & $b$ & $ \p [B\leq b]$ & $\p_{\sf Pois} [B\leq b]$\\
\hline \endfirsthead

\multicolumn{4}{c}%
{{\bfseries \tablename\ \thetable{} -- continued from previous page}} \\
\hline \multicolumn{1}{|c|}{\textbf{Time (s)}} &
\multicolumn{1}{c|}{\textbf{Triple chosen}} &
\multicolumn{1}{c|}{\textbf{Other feasible triples}} \\ \hline 
\endhead

\hline \multicolumn{4}{|r|}{{Continued on next page $\ldots$}} \\ \hline
\endfoot

\endlastfoot

\hline
\multicolumn{4}{|c|}%
{{\bfseries \tablename\ \thetable{} --- continued from previous page}} \\ \hline
  \endhead
 $200$ & $5$ & $2.556545692\times10^{-38}$ & $3.261456422\times10^{-36} $ \\
$200$ & $10$ & $1.450898832\times10^{-32}$ & $1.137687971\times10^{-30} $ \\
$200$ & $50$ & $5.995167631\times10^{-9}$ & $2.401592276\times10^{-8} $ \\
$200$ & $100$ & $0.5265990813$ & $0.5265622074$\\
$100$ & $0$ & $1.322070819\times10^{-23}$ & $1.928749864\times10^{-22} $ \\
$100$ & $5$ & $6.272915577\times10^{-17} $ & $5.567756307\times10^{-16} $ \\
$100$ & $10$ & $ 1.135691814\times10^{-12}$ & $ 6.450152972\times10^{-12} $ \\
$100$ & $15$ & $1.662665039\times10^{-9} $ & $6.357982164\times10^{-9} $\\
$100$ & $20$ & $4.557480806\times10^{-7} $ & $0.000001235187232  $\\
$200$ & $0$ & $ 1.747871252\times10^{-46}$ & $3.720076039\times10^{-44} $\\
$200$ & $5$ & $ 2.556545692\times10^{-38}$ & $3.261456422\times10^{-36} $\\
$200$ & $10$ & $ 1.450898832\times10^{-32}$ & $1.137687971\times10^{-30} $\\
$200$ & $15$ & $ 6.757345217\times10^{-28}$ & $3.340076418\times10^{-26} $\\
$200$ & $20$ & $ 5.962487876\times10^{-24}$ & $1.905558774\times10^{-22} $\\
$500$ & $20$ & $1.240463044\times10^{-84} $ & $1.084188102\times10^{-79} $\\
$500$ & $25$ & $3.140367419\times10^{-79} $ & $1.697380630\times10^{-74} $\\
$500$ & $30$ & $ 2.935666094\times10^{-74}$ & $ 9.912214279\times10^{-70} $ \\
$500$ & $35$ & $1.193158517\times10^{-69} $ & $2.542280876\times10^{-65} $\\
$500$ & $40$ & $2.369596756\times10^{-65} $ & $3.218593843\times10^{-61} $\\
\hline

\multicolumn{4}{c}{$\mathbf{f}=(0.01,0.01,0.01,0.01,0.01)$}\\
\hline
$t$ & $b$ & $\p [B\leq b]$ & $\p_{\sf Pois}[B\leq b]$\\
\hline
$200$ & $5$ & $0.06613951161$ & $0.06708596299$\\
$200$ & $10$ & $0.5830408032$ & $0.5830397512$\\
$200$ & $20$ & $0.9985035184$ & $0.9984117410$\\
$200$ & $50$ & $\approx1$ & $\approx1$\\
\hline
\multicolumn{4}{c}{$\mathbf{f}=(0.2,0.01,0.02,0.03,0.04)$}\\
\hline
$t$ & $b$ & $\p [B\leq b]$ & $ \p_{\sf Pois}[B\leq b]$\\
\hline
$200$ & $5$ & $9.651421837\times10^{-22}$ & $6.180223643\times10^{-20}$\\
$200$ & $10$ & $5.539867010\times10^{-17}$ & $1.744235672\times10^{-15}$\\
$200$ & $20$ & $0.09020056729$ & $0.1076778797$\\
$200$ & $50$ & $0.9999999198$ & $0.9999991415$\\
\hline
\multicolumn{4}{c}{$\mathbf{f}=(0.01,0.01,0.03,0.04,0.05)$}\\
\hline
$t$ & $b$ & $\p [B\leq b]$ & $ \p_{\sf Pois}[B\leq b]$\\
\hline
$200$ & $5$ & $8.312224722\times10^{-8}$ & $1.196952269\times10^{-7}$\\
$200$ & $10$ & $0.00006809921297$ & $0.00008550688580$\\
$200$ & $20$ & $0.06901537242$ & $0.07274102693$\\
$200$ & $50$ & $0.9999582547$ & $0.9999397284$\\
\hline

\multicolumn{4}{c}{$\mathbf{f}=(0.1,0.1,0.1,0.1,0.1)$}\\
\hline





\hline
$t$ & $b$ & $ \p [B\leq b]$ & $\p_{\sf Pois} [B\leq b]$\\
 \hline 

$20$ & $0$ & $0.00002656139888 $ & $0.00004539992984 $ \\
$20$ & $5$ & $0.05757688648 $ & $0.06708596299 $ \\
$20$ & $10$ & $0.5831555123 $ & $0.5830397512 $ \\
$20$ & $15$ & $0.9601094730 $ & $0.9512595983 $ \\
$20$ & $20$ & $0.9991924263 $ & $0.9984117410 $ \\
$40$ & $0$ &  $7.055079108\times10^{-10}$ & $2.061153629\times10^{-9}$ \\
$40$ & $5$ & $0.00003871193246 $ & $ 0.00007190884076$ \\
$40$ & $10$ & $ 0.008071249954$ & $ 0.01081171886 $ \\
$40$ & $15$ & $0.1430754340 $ & $ 0.1565131351 $ \\
$40$ & $20$ & $0.5591747822 $ & $  0.5590925860$ \\
$100$ & $20$ & $4.557480806\times10^{-7}$ & $0.000001235187232 $ \\
$100$ & $25$ & $0.00003540113222 $ & $0.00007160717427 $ \\
$100$ & $30$ & $ 0.001002549708$ & $0.001594027332 $ \\
$100$ & $35$ & $0.01231948910 $ & $0.01621388016 $ \\
$100$ & $40$ & $0.07508928967 $ & $0.08607000083 $ \\

$20$ & $0$ & $0.3660323413 $ & $ 0.3678794412 $ \\
$20$ & $5$ & $ 0.9994654657$ & $ 0.9994058153$ \\
$20$ & $10$ & $0.9999999939 $ & $0.9999999900 $ \\
$20$ & $15$ & $1.000000000 $ & $1.000000000 $ \\
$20$ & $20$ & $1.000000000 $ & $ 1.000000000$ \\
$40$ & $0$ & $ 0.1339796748$ & $ 0.1353352833$ \\
$40$ & $5$ & $0.9839770930 $ & $  0.9834363920$ \\
$40$ & $10$ & $ 0.9999931182$ & $ 0.9999916922$ \\
$40$ & $15$ & $ 0.9999999996 $ & $   1.000000000$ \\
$40$ & $20$ & $ 0.9999999999 $ & $ 1.000000000$ \\
$100$ & $20$ & $ 0.9999999367$ & $ 0.9999999198$ \\
$100$ & $25$ & $ 0.9999999999$ & $ 1.000000001$ \\
$100$ & $30$ & $ 0.9999999999$ & $ 1.000000001$ \\
$100$ & $35$ & $ 0.9999999999$ & $ 1.000000001$ \\
$100$ & $40$ & $ 0.9999999999$ & $ 1.000000001$ \\
\hline
\multicolumn{4}{c}{$\mathbf{f}=(0.02,0.0075,0.0075,0.0075,0.0075)$}\\
\hline
$t$ & $b$ & $\p[B\leq b]$ & $\p_{\sf Pois}[B\leq b]$\\
\hline
$20$ & $0$ & $0.08936904038$ & $ 0.09536916225$ \\
$20$ & $5$ & $0.9712600336$ & $ 0.9672561739$ \\
$20$ & $10$ & $ 0.9999843669$ & $ 0.9999642885$ \\
$20$ & $15$ & $  0.9999999995$ & $ 0.9999999958$ \\
$20$ & $20$ & $  1.000000000$ & $ 1.000000000$ \\
$40$ & $0$ & $ 0.007986825382$ & $ 0.009095277109$ \\
$40$ & $5$ & $   0.6699740391$ & $ 0.6684384858$ \\
$40$ & $10$ & $  0.9927425867$ & $  0.9909776597 $ \\
$40$ & $15$ & $  0.9999835852$ & $   0.9999661876$ \\
$40$ & $20$ & $ 0.9999999935$ & $ 0.9999999715$ \\
$100$ & $20$ & $ 0.9999999935$ & $ 0.9999999715$ \\
$100$ & $25$ & $  0.9999999998$ & $ 1.000000001$ \\
$100$ & $30$ & $ 0.9999999998 $ & $ 1.000000001$ \\
$100$ & $35$ & $ 0.9999999998 $ & $ 1.000000001$ \\
$100$ & $40$ & $ 0.9999999998 $ & $ 1.000000001$ \\
\hline

\end{longtable}
\end{example}

As an example of using the given formula to calculate a confidence interval, suppose we do 200 trials on each of five servers (so there are 1000 trials in total) and we observe 50 failures in total.  Then the resulting confidence interval is $[0,63.29)$.  Suppose we wish to know whether the success probability is at least $\eta=0.9$.  We have $(1-0.9)\times 1000=100$.  This is outside of that interval, and hence we conclude there is enough evidence to reject $H_0$ at the $95\%$ significance level.  However, to test whether the success probability was greater than $0.95$ we see that $(1-0.95)\times1000=50$.  Since $50$ lies within the interval, we conclude there is insufficient evidence to reject $H_0$ at the $95\%$ significance level.




Let $b$ denote the number of incorrect responses we have received from the $c \rho$ challenges given to the provers.  Suppose that $H_0$ is true, so that the expected number of failures is at least $\eta \rho c$.  Based on our approximation, the probability that the number of failures is at most $b$ is at most
\begin{align*}
\sum_{i=0}^b \frac{e^{-\eta \rho c}(\eta \rho c)^i}{i!}.
\end{align*}
If this probability is less than $0.05$, we reject $H_0$ and accept the alternative hypothesis.  However, if the probability is greater than $0.05$, then there is not enough evidence to reject $H_0$ at the $5\%$ significance level, and so we continue to suspect that the file is not being stored appropriately.

We can express this test neatly using a {\em confidence interval}.  We define a $95\%$ upper confidence bound by
\begin{align*}
\lambda_U=\inf\left\{\lambda\Bigg|\sum_{i=0}^b \frac{e^{-\lambda }\lambda^i}{i!}<0.05\right\}.
\end{align*}
This represents the smallest parameter choice for the Poisson distribution for which the probability of obtaining $b$ or fewer incorrect responses is less than $0.05$.  Then $[0,\lambda_U)$ is a $95\%$ confidence interval for the mean number of failures, so we reject $H_0$ whenever $\eta n r$ lies outside this interval.   The value of $\lambda_U$ can be determined easily by exploiting a connection with the chi-squared distribution \cite{Ulm90}: we have that 
\[ \sum_{i=0}^b \frac{e^{-\lambda }\lambda^i}{i!}={\rm Pr}(\chi^2_{2b+2}>2\lambda),\]
 and so the appropriate value of $\lambda_U$ can readily be obtained from a table for the chi-squared distribution.


\section{Optimization Using the Keyed Shacham-Waters Scheme} \label{sec:optimization}
In the last three sections, we gave constructions of $\mspor$ scheme using ramp-schemes, linear secret-sharing schemes, replication codes, and a single-prover-$\por$ system. In this section, we show a specific instantiation of our scheme using the keyed-scheme of Shacham and Waters~\cite{PSU13,SW08} for a single server $\por$ system. 


\subsection{Extension of the Keyed Shacham-Waters Scheme to $\mspor$} \label{sec:extension}
If we instantiate the  $\mathsf{Rep}\mbox{-}\mspor$ scheme (described in~\secref{worst}) with the modified Shacham-Waters scheme of Paterson {\it et al.}~\cite{PSU13}, then we need one key that consists of $n+1$ values in $\F_q$. However, in this case, we do not have any privacy. In particular, we have the following extension of~\corref{rep}. 
\begin{corollary} \label{cor:first}
	Let  $\Pi$ be an $(\eta,0,0,1)\mbox{-}\por$ system of Shacham and Waters~\cite{SW08} with a response code of Hamming distance $\widetilde{\mathsf{d}}$ and the size of challenge space $\gamma$ (where $\widetilde{\mathsf d}$ is given by~\eqnref{d}). Then  $\mathsf{Rep}\mbox{-}\mspor$ instantiated with the Shacham-Waters scheme is an $\mspor$ system with the following properties:
	\begin{enumerate}
		\item Privacy: It is $0$-private.
		\item Security: It is $\paren{\eta,0,1,\rho}$-threshold secure, where $\eta=1 - \frac{\widetilde{\mathsf{d}}(q-1)}{2 \gamma q}$.
		\item Storage Overhead: $\verifier$ needs to store $n+1$ field elements and every $\prover_i$ needs to store $2n$ field elements.
	\end{enumerate}
\end{corollary}

\begin{proof}
	The results follow by combining~\thmref{SW08} with~\corref{rep}. 
\end{proof}

The issue with the $\mathsf{Rep}\mbox{-}\mspor$ scheme is that there is no confidentiality of the file. In what follows, we improve the privacy guarantee of the  $\mspor$ scheme described above. Our starting point would be an instantiation of the $\mathsf{Ramp}\mbox{-}\mspor$ scheme, defined in~\figref{ramp}, with the Shacham-Waters scheme. We then reduce the storage on the $\verifier$ through two steps.

\subsection{Optimized Version of the Multi-server Shacham-Waters Scheme}
We follow two steps to get a  $\mspor$ scheme based on the Shacham-Waters scheme with a reduced storage requirement for the $\verifier$, while improving  the confidentiality guarantee.
\begin{enumerate}
	\item In the first step, stated in~\thmref{first}, we improve the privacy guarantee of the $\mspor$ scheme to get a $\tau_1$-private $\mspor$ scheme (where $\tau_1 <\rho$ is an integer). The $\verifier$ in this scheme has to store $\rho(n+1)$ field elements. When the underlying field is $\F_q$,  the verifier has to store $\rho(n +1) \log q$ bits.
	\item In the second step, stated in~\thmref{optimized_storage}, we reduce the storage requirement of the $\verifier$ from $\rho (n+1)$ to $\tau_1(n +1)$ field elements for some integer $\tau_1 < \rho$ without affecting the privacy guarantee. When the underlying field is $\F_q$,  the verifier has to store $\tau_1(n +1) \log q$ bits.
\end{enumerate}

\paragraph{Step 1.} To improve the privacy guarantee  of~\corref{first} to say, $\tau_1$-private (as per Definition~\ref{defn:privacy}),  we  use a  $\mathsf{Ramp}\mbox{-}\mspor$ scheme and $\rho$ different  keys, where each key consists of $n+1$ values in $\F_q$. The $\verifier$ generates $\rho$ shares of every message block using a ramp scheme, then encodes the shares, and finally computes the tag for each of these encoded shares. 

We follow with more details. Let $m=(m[1],\ldots, m[k])$ be the message. The $\verifier$ computes the shares of every message block $(m[1], \ldots, m[k])$ using a $(\tau_1,\tau_2,\rho)\mbox{-}\mathsf{Ramp}$ scheme. It then encodes all the shares using the encoding scheme of the $\por$ scheme. Let the resulting encoded shares be $M_i[1], \ldots, M_i[n]$ for  $1 \leq i \leq \rho$. In other words, the result of the above two steps are $\rho$ encoded shares, each of which is an $n$-tuple in $(\F_q)^n$. The $\verifier$ now picks random values $a^{(i)},b_1^{(i)}, \ldots, b_n^{(i)}  \in \F_q$ for $1 \leq i \leq \rho$ and computes the tags as follows: 
	$$S_{i}[j]  = b_j^{(i)} + a^{(i)} M_i[j] \qquad \text{for}~1 \leq i \leq \rho, 1 \leq j \leq n.$$ 

The verifier gives $\prover_i$ the tuple of encoded messages $(M_i[1],\ldots, M_i[n])$ and the corresponding tags $(S_i[1], \ldots, S_i[n])$.  We call this scheme the $\mathsf{Basic}\mbox{-}\mspor$ scheme. The following is straightforward from~\thmref{ramp}.

\begin{theorem} \label{thm:first}
	Let  $\Pi$ be an $(\eta,0,0,1)\mbox{-}\por$ scheme  of Shacham and Waters~\cite{SW08} with a response code of Hamming distance $\widetilde{\mathsf{d}}$ and the size of challenge space $\gamma=q^n$ (where $\widetilde{\mathsf d}$ is given by~\eqnref{d}). Let $\mathsf{Ramp}$ be a $(\tau_1,\tau_2,\rho)$-ramp scheme. Then  $\mathsf{Basic}\mbox{-}\mspor$ defined above is an $\mspor$ scheme with the following properties:
	\begin{enumerate}
		\item Privacy: $\mathsf{Basic}\mbox{-}\mspor$ is $\tau_1$-private.
		\item Security: $\mathsf{Basic}\mbox{-}\mspor$ is $\paren{\eta,0,\tau_2,\rho}$-threshold secure, where $\eta=1 - \frac{\widetilde{\mathsf{d}}(q-1)}{2 \gamma q}$.
		\item Storage Overhead: The $\verifier$ needs to store $\rho (n+1)$ field elements and every $\prover_i$ needs to store $2n$ field elements.
	\end{enumerate}
\end{theorem}

In the  construction mentioned above, the $\verifier$ needs to store $\rho (n +1) $ elements of $\F_q$, which is almost the same as the total storage requirements of all the provers.  The same issue was encountered by Paterson {\it et al.}~\cite{PSU13}, where the $\verifier$ has to store as much secret information as the size of the message. This seems to be the general drawback in the unconditionally secure setting. However, in the case of $\mspor$, we can improve the storage requirement of the $\verifier$ as shown in the next  step.

\paragraph{Step 2.} In this step, we  improve the above-described $\mspor$ scheme to achieve considerable reduction on the storage requirement of the $\verifier$. The resulting scheme also provides unbounded audit capability against computationally unbounded adversarial provers, and it also ensures $\tau_1$-privacy.

The main observation that results in the reduction in the storage requirements of the $\verifier$ is the fact that we can partially derandomize the keys generated by the $\verifier$. We use one of the most common techniques in derandomization. The keys in this scheme are generated by  $\tau_1$-wise independent functions.\footnote{A function  is a {\em $\tau_1$-wise independent} function if every subset of $\tau_1$ outputs is independent and equally likely. It should be noted that this does not imply that all the outputs of the function  are mutually independent.}
Our construction works as follows: we pick $n+1$ random polynomials, $f_1(x), \ldots, f_n(x), g(x) \in \F_q[x]$, each of degree at most $\tau_1-1$. Then the $\verifier$ computes the secret key by evaluating the polynomials $f_j(x)$ and $g(x)$ on $\rho$ different values, say
\[ b_j^{(i)} = f_j(i) \qquad \text{and} \qquad a_i=g(i)\] 
for $1 \leq j \leq n,~\text{and}~1 \leq i \leq \rho $.  The $\verifier$ then computes the encoded shares and their corresponding tags  as in $\mathsf{Basic}\mbox{-}\mspor$, i.e., 
	$$S_{i}[j]  = b_j^{(i)} + a^{(i)} M_i[j] \qquad \text{for}~1 \leq i \leq \rho, 1 \leq j \leq n.$$ 

\figref{SW_optimized} is the  formal description of this scheme. 
For the scheme described in~\figref{SW_optimized}, we prove the following result.
\begin{theorem} \label{thm:optimized_storage}
	Let $\mathsf{Ramp = (ShareGen,~Reconstruct)}$ be a $(\tau_1,\tau_2,\rho)$-ramp scheme.  Let  $\Pi$ be a single-prover Shacham-Waters scheme~\cite{SW08} with a response code of Hamming distance $\widetilde{\mathsf{d}}$ and the size of challenge space $\gamma$.  Then $\mathsf{SW}\mbox{-} \mspor$, defined in~\figref{SW_optimized},  is an $\mspor$ system with the following properties:
	\begin{enumerate}
		\item Privacy: $\mathsf{SW}\mbox{-} \mspor$ is $\tau_1$-private.
		\item Security: $\mathsf{SW}\mbox{-} \mspor$ is $\paren{\eta,0,\tau_2,\rho}$-threshold secure, where $\eta=1 - \frac{\widetilde{\mathsf{d}}(q-1)}{2 \gamma q}$.
		\item Storage Overhead: $\verifier$ needs to store $\tau_1(n+1)$ field elements and every $\prover_i$ (for $1 \leq i \leq \rho$) needs to store $2n$ field elements.
	\end{enumerate}
\end{theorem}


\begin{figure} [t]
\begin{center}
\fbox{
\begin{minipage}[l]{6in}
\medskip
{
{\bf Input:} The  $\verifier$ gets a message $m=(m[1], \ldots, m[n])$ as input. Let $\prover_1, \ldots, \prover_\rho$ be the set of $\rho$ provers. Let $q$ be a prime number greater than $\rho$.
\begin{description}
	\item [Initialization Stage:] The $\verifier$ performs the following steps
	\begin{enumerate}
		\item The $\verifier$ choses $n+1$ random polynomials of degree at most $\tau_1-1$, $f_1(x), \ldots, f_n(x),g(x) \in \F_q[x]$ and  a $(\tau_1,\tau_2,\rho )$-ramp scheme $\mathsf{Ramp}=(\mathsf{ShareGen, Reconstruct})$.
		\item For every server $i$, the $\verifier$ does the following:
		\begin{enumerate}
			\item  Compute $\rho$ shares of every message block using the share generation algorithm of $\mathsf{Ramp}$ as follows: $(m_1[j], \ldots, m_\rho[j]) \gets \mathsf{ShareGen}(m[j])$ for $1 \leq j \leq n$.
			\item The $\verifier$  encodes the message  as $e(m_i[j]) = M_i[j]$ for $1 \leq j \leq n, 1 \leq i \leq \rho$.
			\item Compute $b_1^{(i)}=f_1(i), \ldots,b_n^{(i)}= f_n(i), a^{(i)}=g(i) $. 
			\item Compute the tag $S_{i}[j]= b_j^{(i)} + a^{(i)} M_{i}[j]$ for $1 \leq j \leq n$. 
			\end{enumerate}
		\item The $\verifier$ gives $\set{(M_{i}[j], S_{i}[j])}_{1 \leq j \leq n}$ to $\prover_i$.
	\end{enumerate}
	\item [Challenge Phase:] During the audit phase, $\verifier$ picks a prover, $\prover_i$, and runs the challenge-response algorithm of a single-server Shacham-Waters scheme. It computes the corresponding keys by computing the random polynomials chosen during the set up phase.
 \end{description}
}
\end{minipage}
}\caption{$\mspor$ Using Optimized Shacham-Waters scheme ($\mathsf{SW}\mbox{-} \mspor$).} \label{fig:SW_optimized} 
\end{center}
\end{figure}

\begin{proof}
	 The privacy guarantee of $\mathsf{SW\mbox{-}}\mspor$ is straightforward from the secrecy property of the underlying ramp scheme. 
	 
For the security guarantee, we have to show an explicit construction of $\ext$, that on input  proving algorithms $\cP_1, \ldots , \cP_\rho$, outputs $m$ if $\suc(\cP_i) > \eta$ for at least $\tau_2$ proving algorithms. However, there is a subtle issue that we have to deal with before using the proof of~\thmref{ramp}, because of the relation between every message and tag pair. It was noted by Paterson {\it et al.}~\cite{PSU13} that, if the adversarial prover learns the secret key, then it can break the $\por$ scheme.  We first argue that a set of $\tau_1$ colluding provers cannot have an undue advantage from exploiting the linear structure of the message-tag pairs.


We now prove that any set of $\tau_1$ provers do not learn anything about the keys generated using $n+1$ polynomials of degree at most $\tau_1-1$. The idea is very similar to the single-prover case. Paterson {\it et al.}~\cite{PSU13} noted that in the single prover case, for an $n$-tuple encoded message, the key is a tuple of $n+1$ uniformly random elements $(a,b_1, \ldots, b_n)$ in $\F_q$. Further, from the point of view of a prover, there  are $q$ possible keys -- the value of $a$ determines the $n$-tuple $(b_1, \ldots, b_n)$ uniquely, but $a$ is completely undetermined. In the $\mspor$ case, we have $\rho$ keys. Each prover in a given set of $\tau_1$ provers has $q$ possible keys, as discussed above. However, it is conceivable that they can use their collective knowledge to learn something about the keys.  In what follows, we show that they cannot determine any additional information about their keys by combining the information they hold collectively. 

Let $I=\set{i_1,\ldots, i_{\tau_1}}$ be the indices of any arbitrary set of $\tau_1$ provers. Let $S_i$ denote the set of possible keys for $\prover_i$, for $i \in I$. Consider any list of $\tau_1$ keys $(K_{i_1}, K_{i_2}, \ldots, K_{i_{\tau_1}})$. Recall that $K_i$ (for $i \in I$) has the form $\paren{a^{(i)}, b^{(i)}_1, \ldots, b^{(i)}_n}$, where $a^{(i)}$  and $b^{(i)}_j$ (for $1\leq j \leq n$) are generated by random polynomials of degree $\tau_1$.  We first consider $a^{(i)}$ (for $i \in I$). Note that the vector $\paren{b^{(i)}_1, \ldots, b^{(i)}_n}$ is defined uniquely by $a^{(i)}$ and the  set of all encoded message-tag pairs. We have already shown that any set of $\tau_1$ provers cannot learn anything about the random polynomial $g(x)$ used to generate the $a^{(i)}$ for all $i \in I$. We use the following well-known fact to show the any set of $\tau_1$ provers do not learn any additional information about the keys.

\begin{fact} \label{fact:twise}
Let $t>0$ be an integer, let $q$ be a prime number, and let $\F_q$ be a finite field. Let $h_0,h_1, \ldots, h_{t-1} \in \F_q$ be random elements picked uniformly at random. Define $h(x)=\sum_{i=0}^{t-1} h_i x^i$ for all $\alpha \in \F_q$. Then,
\begin{align}
 \p \sparen{h\paren{x_1} = y_1 \wedge \ldots \wedge h\paren{x_{\tau}} = y_{t}} = \prod_{i=1}^{t} \p \sparen{h(x_{\alpha_i}) = y_i }. \label{eq:twise}
 \end{align}
 Since $h(x)$ is uniformly distributed in $\F_q$, the probability computed in~\eqnref{twise} is actually equal to $q^{-t}$.\end{fact}

By construction, $g(x)$ is a random polynomial of degree at most $\tau_1-1$. Fact~\ref{fact:twise} then implies that any combination of $\set{a^{(i)}}_{i \in I}$ is equally likely. A similar argument, with the  $a^{(i)}$'s replaced by  the $b^{(i)}_j$'s (for all $i \in I$ and $1\leq j \leq n$) and the polynomial $g(x)$ replaced by $f_j(x)$ (for $1 \leq j \leq n$), gives that all  set of $\tau_1$ keys are equally likely.  In other words, the set of provers in the set $I$ cannot determine any additional information about their keys by combining the information they hold collectively. 

We now complete the security proof by describing an $\ext$ that outputs the file if $\tau_2$ provers succeed with high enough probability. 
The description of the $\ext$ and its analysis is same as that of~\thmref{ramp}. We give it for the sake of completeness.
	\begin{enumerate}
		\item $\ext$ chooses $\tau_2$ provers and runs the extraction algorithm of the underlying single-server $\por$ system on each of these provers. In the end, it outputs $\widehat{M}_{i_j}$ for the corresponding provers $\prover_{i_j}$. It defines $\mathcal{S} \gets \{\widehat{m}_{i_1}, \ldots, \widehat{m}_{i_{\tau_2}} \}$. Note that the $\ext$ of the underlying $\por$ scheme has already computed $e^{-1}$ on the set $\set{ \widehat{M}_{i_1}, \ldots, \widehat{M}_{i_{\tau_2}} }$.
		\item $\ext$ invokes the $\mathsf{Reconstruct}$ algorithm of the underlying ramp scheme with the  elements of $\widetilde{\mathcal{S}}$ to compute $m'$. 
	\end{enumerate}	  
	 Now note that the $\verifier$ interacts with every $\prover_i$ independently. We know from the security of the underlying $\por$ scheme of Shacham-Waters that there is an extractor that always outputs the encoded message whenever $\suc_\mathsf{avg}(\cP_i) \geq \eta.$ Therefore, if all the $\tau_2$ chosen proving algorithms succeed with probability at least $\eta$ over all possible keys, then the set $\mathcal{S}$  will have  $\tau_2$ correct shares. From the correctness of the $\mathsf{Reconstruct}$ algorithm and $e^{-1}(\cdot)$, we know that the message output in the end by the $\ext$ will be the message $m$.

	 For the storage requirement, the $\verifier$ has to  store the coefficients of all the random polynomials $f_1(x), \ldots, f_n(x),g(x)$, which amounts to a total of $\tau_1(n+1) = \tau_1 n +n$ field elements.  	
	 \end{proof}

\section{Conclusion}
In this paper, we studied $\por$ systems when multiple provers are involved ($\mspor$). We motivated and defined the security of $\mspor$ in the worst-case~(Definition~\ref{defn:worst}) and the average-case~(Definition~\ref{defn:average}) settings, and extended the hypothesis testing techniques used in the single-server setting~\cite{PSU13} to the multi-server setting. We also motivated the study of confidentiality of the outsourced message. We gave $\mspor$ schemes which are secure under both these security definitions and provide reasonable confidentiality guarantees even when there is no restriction on the computational power of the servers. 
 In the end of this paper, we looked at an optimized version of $\mspor$ system when instantiated with the unconditionally secure version of the Shacham-Waters scheme~\cite{SW08}. We exhibited that, in the multi-server setting with computationally unbounded provers, one can overcome the limitation that the verifier needs to store as much secret information as the provers.

\bibliographystyle{plain}
\bibliography{por1}

\begin{table}[htb]
\label{notation.tab}
\begin{center}
\begin{tabular}{l|l}
$c$ & challenge \\ \hline
$C^\bot$ & dual of a code $C$ \\ \hline
$\mathsf{d}^*$ & distance of the response code \\ \hline
$\mathsf{d}$ & distance of a codeword \\ \hline
$\mathsf{d}^\bot$ & dual distance of a code \\ \hline
$\mathsf{dist} $ & hamming distance between two vectors \\ \hline
$\mathbf{G}$ & generator matrix of a code\\ \hline
$k$ & length of a message \\ \hline
$K$ & key (in a keyed scheme) \\ \hline
$\ell$ & number of message-blocks \\ \hline
$m$ & message \\ \hline
$m[i]$ & $i$-th message block \\ \hline
$\widehat{m}$ & message outputted by the $\ext$ \\ \hline
$\mathcal{M}$ & message space \\ \hline
$M$ & encoded message \\ \hline
$M[i]$ & $i$-th encoded message \\ \hline
$M_j[i]$ & $i$-th encoded message on $\prover_j$\\ \hline
$\mathcal{M}^*$ & encoded message space \\ \hline
$n$ & number of provers  \\ \hline
$N$ & codeword length  \\ \hline
$\mathcal{P}_i$ & proving algorithm of $i$-th $\mathsf{Prover}$ \\ \hline
$q$ & order of underlying finite field\\ \hline
$r$ & response \\ \hline
$r^M$ & response vector for encoded message $M$\\ \hline
$S$ & tag (in a keyed scheme) \\ \hline
$\mathsf{succ}(\mathcal{P})$ & success probability of proving algorithm \\ \hline
$\mathcal{R}^*$ & response code \\ \hline
$\Gamma$ & challenge space \\ \hline
$\gamma$ & number of possible challenges \\ \hline
$\Delta$ & response space \\ \hline
$\varphi$ & column sparsity of a matrix \\ \hline
$\rho$  & response function \\ \hline
$\tau$ & privacy threshold \\ \hline
$\zeta$ & row sparsity of a matrix \\ \hline
\end{tabular}
\caption{Notation used in this paper}
\end{center}
\end{table}

\end{document}